\documentclass[preprintnumbers,article,amsmath,amssymb,floatfix,10pt,prd,onecolumn,
superscriptaddress,nofootinbib]{revtex4-2}
\usepackage{bm}
\usepackage{amsfonts}
\usepackage{latexsym}
\usepackage[latin1]{inputenc}
\usepackage{graphicx}
\usepackage{amsmath}
\usepackage{palatino}
\usepackage{mathpazo}
\usepackage{textcomp}
\usepackage{float}
\usepackage{booktabs}
\usepackage{dcolumn}
\usepackage{ragged2e}
\usepackage{hyperref}
\hypersetup{colorlinks,citecolor=blue}
\hypersetup{colorlinks=true,linkcolor=red,filecolor=magenta,    urlcolor=blue}
\usepackage{amsmath}
\usepackage{xcolor}
\usepackage{orcidlink}
\usepackage{epsfig}
\usepackage[caption=false]{subfig}
\usepackage{commath}
\usepackage{amsmath,amssymb,amsfonts,amsthm,indentfirst,setspace}
\usepackage{xcolor}\usepackage{ulem}
\usepackage{graphicx}
\usepackage{float}
\newtheorem{theorem}{Theorem}[section]

\newtheorem{lemma}{Lemma}[section]
\newtheorem{corollary}{Corollary}[section]
\newtheorem{remark}{Remark}[section]
\newcommand{\be}{\begin{equation}}
\newcommand{\ee}{\end{equation}}
\newcommand{\bea}{\begin{eqnarray}}
\newcommand{\eea}{\end{eqnarray}}
\newcommand{\eeas}{\end{eqnarray*}}
\newcommand{\beas}{\begin{eqnarray*}}

\def\jnl@style{\it}
\def\aaref@jnl#1{{\jnl@style#1}}

\def\aaref@jnl#1{{\jnl@style#1}}

\def\aj{\aaref@jnl{AJ}}                   
\def\apj{\aaref@jnl{ApJ}}                 
\def\apjl{\aaref@jnl{ApJ}}                
\def\apjs{\aaref@jnl{ApJS}}               
\def\apss{\aaref@jnl{Ap\&SS}}             
\def\aap{\aaref@jnl{A\&A}}                
\def\aapr{\aaref@jnl{A\&A~Rev.}}          
\def\aaps{\aaref@jnl{A\&AS}}              
\def\mnras{\aaref@jnl{Mon.~Not.~Roy.~Astron.~Soc.}}             
\def\prd{\aaref@jnl{Phys.~Rev.~D}}        
\def\prc{\aaref@jnl{Phys.~Rev.~C}}  
\def\prl{\aaref@jnl{Phys.~Rev.~Lett.}}    
\def\qjras{\aaref@jnl{QJRAS}}             
\def\skytel{\aaref@jnl{S\&T}}             
\def\ssr{\aaref@jnl{Space~Sci.~Rev.}}     
\def\zap{\aaref@jnl{ZAp}}                 
\def\nat{\aaref@jnl{Nature}}              
\def\aplett{\aaref@jnl{Astrophys.~Lett.}} 
\def\apspr{\aaref@jnl{Astrophys.~Space~Phys.~Res.}} 
\def\physrep{\aaref@jnl{Phys.~Rep.}}      
\def\physscr{\aaref@jnl{Phys.~Scr}}       
\def\commat{\aaref@jnl{Comm.~Math.~Phys.}}              
\def\science{\aaref@jnl{Science}}               
\def\cqg{\aaref@jnl{Classical Quant.~Grav.}}            
\def\jpcs{\aaref@jnl{JPCS}}                                     
\def\ijmpd{\aaref@jnl{Int.~J.~Mod.~Phys.~D}}                    
\def\grg{\aaref@jnl{Gen.~Relat.~Gravit.}}               
\def\rpp{\aaref@jnl{Rep.~Prog.~Phys.}}          
\def\npa{\aaref@jnl{Nucl.~Phys.~A}}        
\def\lrr{\aaref@jnl{Living Rev.~Rel.}}                   
\def\jcap{\aaref@jnl{J.~Cosmology Astropart.~Phys.}}    
\def\rmp{\aaref@jnl{Rev.~Mod.~Phys.}}   
\def\epjc{\aaref@jnl{Eur.~Phys.~J.~C}} 
\def\plb{\aaref@jnl{~Phy.~Lett.~B}} 
\def\mpla{\aaref@jnl{Mod.~Phy.~Lett.~A}} 
\def\arxiv{\aaref@jnl{arxiv.org}}

\allowdisplaybreaks[1]

\begin{document}
\color{black}       
\title{HOW A PROJECTIVELY FLAT GEOMETRY REGULATES $F(R)$-GRAVITY THEORY?}

\author{Tee-How Loo\orcidlink{0000-0003-4099-9843}}
\email{looth@um.edu.my}
\affiliation{Institute of Mathematical Sciences,
University of Malaya,
50603 Kuala Lumpur,
Malaysia}

\author{Avik De\orcidlink{0000-0001-6475-3085}}
\email{de.math@gmail.com}
\affiliation{Department of Mathematical and Actuarial Sciences, Universiti Tunku Abdul Rahman,\\
Jalan Sungai Long,
43000 Cheras,
Malaysia}

\author{Sanjay Mandal\orcidlink{0000-0003-2570-2335}}
\email{sanjaymandal960@gmail.com}
\affiliation{Department of Mathematics, Birla Institute of Technology and
Science-Pilani,\\ Hyderabad Campus, Hyderabad-500078, India.}

\author{P.K. Sahoo\orcidlink{0000-0003-2130-8832}}
\email{pksahoo@hyderabad.bits-pilani.ac.in}
\affiliation{Department of Mathematics, Birla Institute of Technology and
Science-Pilani,\\ Hyderabad Campus, Hyderabad-500078, India.}
%
\date{\today}
\begin{abstract}
In the present paper we examine a projectively flat spacetime solution of $F(R)$-gravity theory. It is seen that once we deploy projective flatness in the geometry of the spacetime, the matter field has constant energy density and isotropic pressure. We then make the condition weaker and discuss the effects of projectively harmonic spacetime geometry in $F(R)$-gravity theory and show that the spacetime in this case reduces to a generalised Robertson-Walker spacetime with a shear, vorticity, acceleration free perfect fluid with a specific form of expansion scalar presented in terms of the scale factor. Role of conharmonic curvature tensor in the spacetime geometry is also briefly discussed. Some analysis of the obtained results are conducted in terms of couple of $F(R)$-gravity models.
\end{abstract}

\maketitle

\date{\today}
\section{\textbf{Introduction}}
In 1918, H. Weyl constructed, what today is known as the Weyl geometry \cite{weyl1} and in 1921, in a subsequent paper \cite{weyl}, he defined a projective curvature tensor $P^l{}_{ijk}$ as a projectively invariant component of the Riemann curvature tensor. Suppose in an $n$-dimensional (pseudo-)Riemannian manifold $M^n$, we consider a torsion-less affine connection $\Gamma^l_{\:\:jk}$, with the corresponding curvature tensor 
$R^l{}_{ijk}$ and Ricci tensor $R_{ik}$ respectively. The projective curvature tensor is given by
\begin{align}\label{eqn:01}
P^l{}_{ijk}=R^l{}_{ijk}+L_{ji}\delta^l_k-L_{ki}\delta^l_j+(L_{jk}+L_{kj})\delta^l_j,
\end{align}  
where 
\[
L_{ik}=\frac{nR_{ki}+R_{ik}}{n^2-1}.
\]
A curve $x^l(t)$ in $M$ is said to be path induced by the affine connection $\Gamma^l{}_{jk}$ if it satisfies the following system of differential equations
\[
\ddot{x}^l+\dot x^j\dot x^k\Gamma^l{}_{jk}=\alpha\dot x^l,
\] 
where $\alpha(t)$ is a scalar tensor and $\dot x^l=dx^l/dt$.
Suppose there is another torsion-less affine connection $\tilde \Gamma^l{}_{jk}$ on $M$ such that both $\Gamma^l{}_{jk}$ and $\tilde \Gamma^l{}_{jk}$ 
give the same system of paths. Then  we must have the relation
\[
\tilde{\Gamma}^l_{\:\:jk}-\Gamma^l_{\:\:jk}=p_k\delta^l_j+p_j\delta^l_k,
\]
for some 1-form $p_i$. 
We say that $\Gamma^l{}_{jk}$ and $\tilde\Gamma^l{}_{jk}$ are projectively related. 
Furthermore, if $\tilde P^l_{ijk}$ denote the projective curvature tensor corresponding to $\tilde \Gamma^l{}_{jk}$, then the following relation hold.
\[
\tilde P^l{}_{ijk}=P^l{}_{ijk}.
\]
This means that the project curvature tensor is invariant under a projective change of $\Gamma^l{}_{jk}$.
In particular, if $\Gamma^l{}_{jk}$ is the Levi-Civita connection, then (\ref{eqn:01}) becomes
\be 
P_{hijk}=R_{hijk}+\frac{R_{ij}g_{hk}-R_{ik}g_{jh}}{n-1}.
\label{proj}\ee
It is known that the projective curvature tensor coincides with the Weyl conformal curvature tensor if and only if the underline metric is Einstein \cite{projconf}. Projective curvature symmetry which
preserves the Weyl projective curvature tensor carries significant interest in the standard theory of gravity \cite{projsym} governed by Einstein's field equations (EFE):
\begin{equation}\label{efe}
R_{ij}-\frac{R}{2}g_{ij}=\kappa^2T_{ij}.
\end{equation}

With the 1998 discovery of the accelerating expansion of the universe using type Ia supernovae, and due to the limited knowledge about the character of dark energy, an alternative theory of gravity was resurfaced by modifying the Einstein-Hilbert action by replacing the Ricci scalar $R$ with an arbitrary function $F(R)$. We should mention here that the Weyl geometry mentioned in the beginning of the present article was primarily proposed as an alternative of Einstein's gravity theory and a unified theory of gravitation
and electromagnetism \cite{weyl-survey}. Some of the initial $F(R)$-models, for example, $F(R)=R+\alpha R^2$, $(\alpha>0)$ proposed by Starobinsky \cite{star} and $F(R)=R-\mu^4/R$, $(\mu>0)$ proposed by Carroll et al. \cite{carroll} were aimed at explaining respectively the cosmic inflation without using any scalar field and late-time acceleration as a pure gravitational effect without using any dark energy and these models with all their limitations, still managed to popularise the $F(R)$-models, in general. But it also showed how non-trivial the process of finding a viable, stable $F(R)$-model is. For a detailed survey on $F(R)$-gravity theory see \cite{f(R)-survey} and the references therein. 

The action term 
\[S=\frac{1}{2\kappa^2}\int F(R) \sqrt{-g}d^4x +\int L_m\sqrt{-g}d^4x,\]
gives us the well-known field equations of the $F(R)$-gravity as
\be 
F_R(R)R_{ij}-\frac{1}{2}F(R)g_{ij}+(g_{ij}\Box-\nabla_i\nabla_j)F_R(R)=\kappa^2T_{ij},\label{frgr1}	
\ee
where $\Box=\nabla^k\nabla_k$, $L_m$ is the matter Lagrangian, and 
\[T_{ij}=-\frac{2}{\sqrt{-g}}\frac{\delta(\sqrt{-g}L_m)}{\delta g^{ij}},\] is the stress-energy tensor. We consider a perfect fluid type stress-energy tensor given by 
\be T_{ij}=ph_{ij}+\rho u_iu_j,
\label{pf}\ee
where $p$ and $\rho$ are the isotropic pressure and energy density, respectively related by an equation of state (EoS) $\omega$ as $p=\omega \rho$. $h_{ij}=g_{ij}+u_iu_i$ denote the orthogonl projector.

The field equations (\ref{frgr1}) can be rewritten as an effective theory 
\be R_{ij}-\frac{R}{2}g_{ij}=\frac{\kappa^2}{F_R(R)}T_{ij}^{\text{eff}} \notag
\ee
where 
\beas\label{T^eff}
\kappa^2T_{ij}^{\text{eff}}=\kappa^2T_{ij}+\frac{F(R)-RF_R(R)}{2}g_{ij}+\left(\nabla_i\nabla_j-g_{ij}\Box\right)F_R(R).
\eeas
We often use the term 
$$
\kappa^2 T_{ij}^{\text{curv}}=\frac{F(R)-RF_R(R)}{2}g_{ij}+\left(\nabla_i\nabla_j-g_{ij}\Box   \right)F_R(R).
$$
Different type of symmetries of the spacetimes, although initiated probably for plain mathematical purpose, are useful in the gravity theories to find exact solutions of the field equations. Finding a solution of Einstein's field equations or of any modified gravity theories like the present $F(R)$-theory means finding a metric $g_{ij}$ given in terms of elementary functions (the geometry) where the physics part, the $T_{ij}$ is guided by physical laws such as energy conditions etc. This is already very complicated considering the higher order non-linear differentiation of the metric tensor involved. And simply a local solution of the set of differential equations is not enough, either, the global solution or topological properties are to be taken into account. So traditionally we always impose some kind of symmetry into the metric, the Riemannian tensor, the Ricci tensor and alike. Then once the solution is found, we study the physical interpretation. Spherical symmetric, axisymmetric static, non-static spherically symmetric, plane symmetric etc are some of the popular outcome of this. In the present paper we impose two constraints on the spacetime, namely, the projectively flatness and null divergence of the projective curvature tensor to investigate the role of the projective curvature tensor in $F(R)$-gravity theories. 

After the introduction, we discuss the field equation of $F(R)$-gravity in a projectively flat spacetime. In the next section we analyse our findings using two popular models of $F(R)$-gravity. In the next section, we investigate the effect of divergence free projective curvature tensor followed by a discussion.
 
\section{\textbf{Projectively flat spacetime}}
We begin this section with a lemma which is useful in dealing with the modified $F(R)$-gravity as an effective counterpart of the standard theory of gravity, in general. 
\begin{lemma}\label{lem1}
In a perfect fluid spacetime solution of the $F(R)$-gravity theories with constant Ricci scalar $R$, the effective pressure $p^{\text{eff}}$ and energy density $\rho^{\text{eff}}$ satisfy the following relation:
\begin{align*}
 p^{\text{eff}}=p+\frac{F(R)-RF_R(R)}{2\kappa^2},\\\rho^{\text{eff}}=\rho-\frac{F(R)-RF_R(R)}{2\kappa^2}.
\end{align*}
\end{lemma}
\begin{proof}
The field equations (\ref{frgr1}) of $F(R)$-theories of gravity in case of a constant $R$ reduces to the form
\be R_{ij}-\frac{R}{2}g_{ij}=\frac{\kappa^2}{F_R(R)}T_{ij}+\frac{F(R)-RF_R(R)}{2F_R(R)}g_{ij}.\label{fr}\ee
Using the perfect fluid form (\ref{pf}), simple calculations produce the result. 
\end{proof}

\begin{theorem}\label{thm1}
In any projectively flat perfect fluid spacetime solution of $F(R)$-gravity theory, the energy density $\rho$ and isotropic pressure $p$ are constants and separately satisfy certain relations with the gravity sector as follows: 
$$p=\frac{RF_R(R)}{4\kappa^2}-\frac{F(R)}{2\kappa^2}, \qquad \rho=\frac{F(R)}{2\kappa^2}-\frac{RF_R(R)}{4\kappa^2}.$$ 
\end{theorem}
\begin{proof}
In a 4-dimensioanl projectively flat spacetime, the projective curvature tensor vanishes, therefore from (\ref{proj}) we can express the Riemannian curvature tensor as
\be 3R_{hijk}+R_{ij}g_{hk}-R_{ik}g_{jh}=0.\label{riem}\ee
Using the identities of Riemannian curvature tensor, $R_{hijk}+R_{ihjk}=0$, from (\ref{riem}) we can readily obtain
\beas R_{ij}g_{hk}-R_{ik}g_{jh}=R_{hk}g_{ij}-R_{hj}g_{ik}\eeas 
which on contraction of $h$ and $k$ gives
\be R_{ij}=\frac{R}{4}g_{ij}.\label{rij}\ee
Again, covariantly differentiating (\ref{riem}) we get
\be 3\nabla^hR_{kijh}=\nabla_jR_{ik}-\nabla_kR_{ij}\label{v1}\ee
Using the contracted Bianchi's second identity $\nabla^mR_{jklm}=\nabla_jR_{kl}-\nabla_kR_{jl}$ in (\ref{v1}) we conclude that the Ricci curvature is of Codazzi type
\be \nabla_i R_{jk}-\nabla_jR_{ik}=0.\label{v2}\ee
Contracting $j$ and $k$ in (\ref{v2}) and using the identity $\nabla^kR_{ik}=\frac{1}{2}\nabla_iR$ we obtain
\be \nabla_iR=0.\label{r}\ee
(\ref{rij}) and (\ref{r}) together imply 
\be \nabla_kR_{ij}=0.\label{a1}\ee

Due to (\ref{rij}), for constant $R$ the field equations (\ref{fr}) reduces to
\be -\frac{R}{4}g_{ij}=\left[\frac{\kappa^2}{F_R(R)}p+\frac{F(R)-RF_R(R)}{2F_R(R)}  \right]g_{ij}+\frac{\kappa^2}{F_R(R)}(p+\rho)u_iu_j.\label{a2}\ee
Contracting $i$ and $j$ in (\ref{a2}) we get
\be R=\frac{3\kappa^2 p}{F_R(R)}-\frac{\kappa^2 \rho}{F_R(R)}+\frac{2F(R)}{F_R(R)}.\label{a3}\ee
Again transvecting (\ref{a2}) by $u^j$ and after some calculations we get
\be R=\frac{2F(R)}{F_R(R)}-\frac{4\kappa^2 \rho}{F_R(R)}.\label{a4}\ee
(\ref{a3}) and (\ref{a4}) together imply
\beas p=\frac{RF_R(R)}{4\kappa^2}-\frac{F(R)}{2\kappa^2},\eeas
and
\beas \rho=\frac{F(R)}{2\kappa^2}-\frac{RF_R(R)}{4\kappa^2}.\eeas
\end{proof}
\begin{remark}
In a projectively flat spacetime, Ricci scalar $R$ is constant.
\end{remark}
Using Lemma \ref{lem1} and Theorem \ref{thm1} it is straightforward to obtain the following result:
\begin{corollary}
In any projectively flat  perfect fluid spacetime solution of $F(R)$-gravity theory, the effective energy density $\rho^{\text{eff}}$ and isotropic pressure $p^{\text{eff}}$ are constants and given by the expressions: 
\beas p^{\text{eff}}=-\frac{RF_R(R)}{4\kappa^2}, \qquad \rho^{\text{eff}}=\frac{RF_R(R)}{4\kappa^2}.\eeas
\end{corollary}


\section{\textbf{Graphical analysis of the scenario}}

Energy conditions are one of the greatest tools to examine the cosmological models' self-stability, which are generally derived from the well-known Raychaudhuri equation \cite{Raychaudhuri/1995,R2,R3}. It also helps us describe the space-time curve's geometrical behavior such as spacelike, timelike and lightlike, and gives new insights into dreadful singularities \cite{Moraes/2017,M2}. The Raychaudhuri equations can be written if the following forms
\begin{equation}
\label{16}
\frac{d\theta}{d\tau}=-\frac{1}{3}\theta^2-\sigma_{\mu\nu}\sigma^{\mu\nu}+\omega_{\mu\nu}\omega^{\mu\nu}-R_{\mu\nu}u^{\mu}u^{\nu}\,,
\end{equation}
\begin{equation}
\label{17}
\frac{d\theta}{d\tau}=-\frac{1}{2}\theta^2-\sigma_{\mu\nu}\sigma^{\mu\nu}+\omega_{\mu\nu}\omega^{\mu\nu}-R_{\mu\nu}n^{\mu}n^{\nu}\,,
\end{equation}
where $\theta$ is the expansion factor, $n^{\mu}$ is the null vector, and $\sigma^{\mu\nu}$ and $\omega_{\mu\nu}$ are, respectively, the shear and the rotation associated with the vector field $u^{\mu}$. For attractive gravity, equations \eqref{16}, and \eqref{17} satisfy the following conditions
\begin{align*}
\label{18}
R_{\mu\nu}u^{\mu}u^{\nu}\geq0\,,\\
 R_{\mu\nu}n^{\mu}n^{\nu}\geq0\,.
\end{align*}
 Therefore, if we are working with a perfect fluid matter distribution and collaborating with the work from Capozziello et al., \cite{capozziello}, the energy conditions recovered from standard GR are
\begin{itemize}
\item Strong energy conditions (SEC) if  $\rho^{eff}+3p^{eff}\geq 0\,$;

\item Weak energy conditions (WEC) if  $\rho^{eff}\geq 0, \rho^{eff}+p^{eff}\geq 0\,$;

\item Null energy condition (NEC) if  $\rho^{eff}+p^{eff}\geq 0\,$;

\item Dominant energy conditions (DEC) if $\rho^{eff}\geq 0, |p^{eff}|\leq \rho\,$.
\end{itemize}

Using Lemma 2.1, we have rewritten the energy conditions as follows

\begin{itemize}
\item Strong energy conditions (SEC) if  $\rho+3p\geq 0\,$ with $R F_R(R)-F(R)\geq 0$;

\item Weak energy conditions (WEC) if  $\rho\geq 0, \rho+p\geq 0\,$ with $-R F_R(R)+F(R)\geq 0$;

\item Null energy condition (NEC) if  $\rho+p\geq 0\,$;

\item Dominant energy conditions (DEC) if $\rho\geq 0, |p|\leq \rho\,$.
\end{itemize}

 Now, in the following subsections, we discuss the energy conditions for two specific functions of Ricci scalar $R$ under this formulation. Furthermore, we would like to note here that under projectively flat spacetime, $R$ is constant. Therefore, the continuity equation follows directly i.e., $\nabla^iT_{ij}=0$.
 
\subsection{Model-1: $F(R)=R-\alpha  \left(1-e^{-\frac{R}{\alpha }}\right)$}

Here we presume $F(R)$ as a exponential function with free parameter $\alpha$. Therefore, the energy density and pressure can be written as
\begin{equation}\label{23}
\rho=\frac{1}{4} e^{-\frac{R}{\alpha }} \left(2 \alpha +R e^{R/\alpha }-2 \alpha  e^{R/\alpha }+R\right),
\end{equation}

\begin{equation}\label{24}
p=-\frac{1}{4} e^{-\frac{R}{\alpha }} \left(2 \alpha +R e^{R/\alpha }-2 \alpha  e^{R/\alpha }+R\right).
\end{equation}
Now, using Eqn. \eqref{23} and \eqref{24}, one can discuss the energy conditions with this setup. The profiles of all energy conditions are depicted in Fig 1, 2, 3 except WEC. In this setup, $\rho+p$ reduces to zero. The energy density is positive throughout the evolution of Ricci Scalar $R>1$ and $\alpha>0$. From Figure 1, one can observe that the energy density is high for higher values of $R$, which suggests the energy density of our universe was high in the early time, and later it decreases as Ricci scalar $R$ decreases with time. Also, NEC is part of WEC. Therefore, WEC and NEC are satisfied. The profile of DEC is presented in Fig. 2, and it takes its value in the positive range. SEC is violated, and this result indicates the late-time acceleration of the universe \cite{sm1, sa2,sm2,sm3,sa1,sm4}. In addition, the equation of state parameter, $\omega=p/\rho=-1$ for this formulation. Moreover, All of the results are compatible with the $\Lambda$CDM model \cite{planck/2018}.

\begin{figure}[H]
\begin{center}
\includegraphics[scale=0.4]{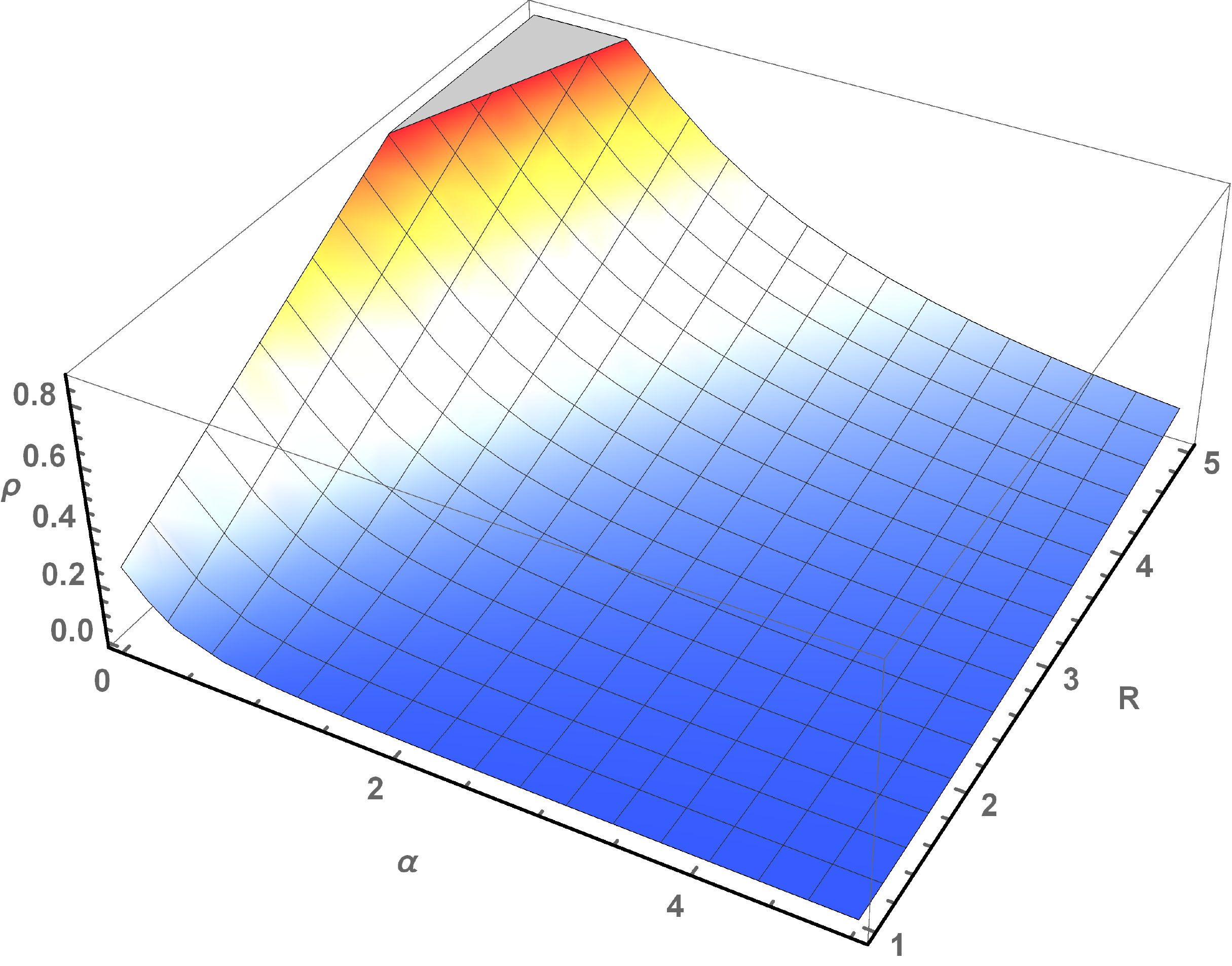}
\caption{Evolution of energy density with respect to $R$ and $\alpha$ for model-1.}
\end{center}
\end{figure}

\begin{figure}[H]
\begin{center}
\includegraphics[scale=0.4]{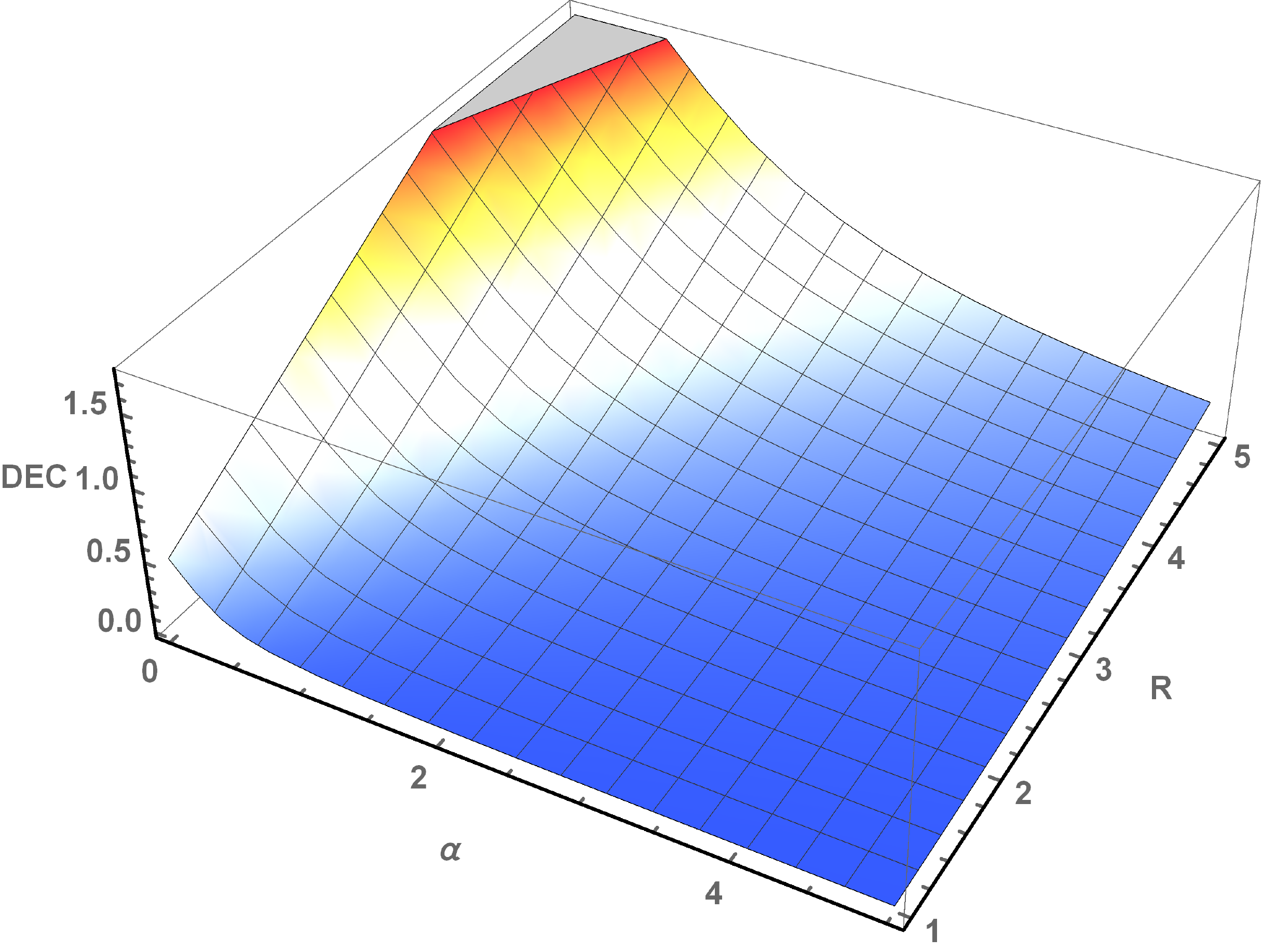}
\caption{Evolution of DEC with respect to $R$ and $\alpha$ for model-1.}
\end{center}
\end{figure}
\begin{figure}[H]
\begin{center}
\includegraphics[scale=0.4]{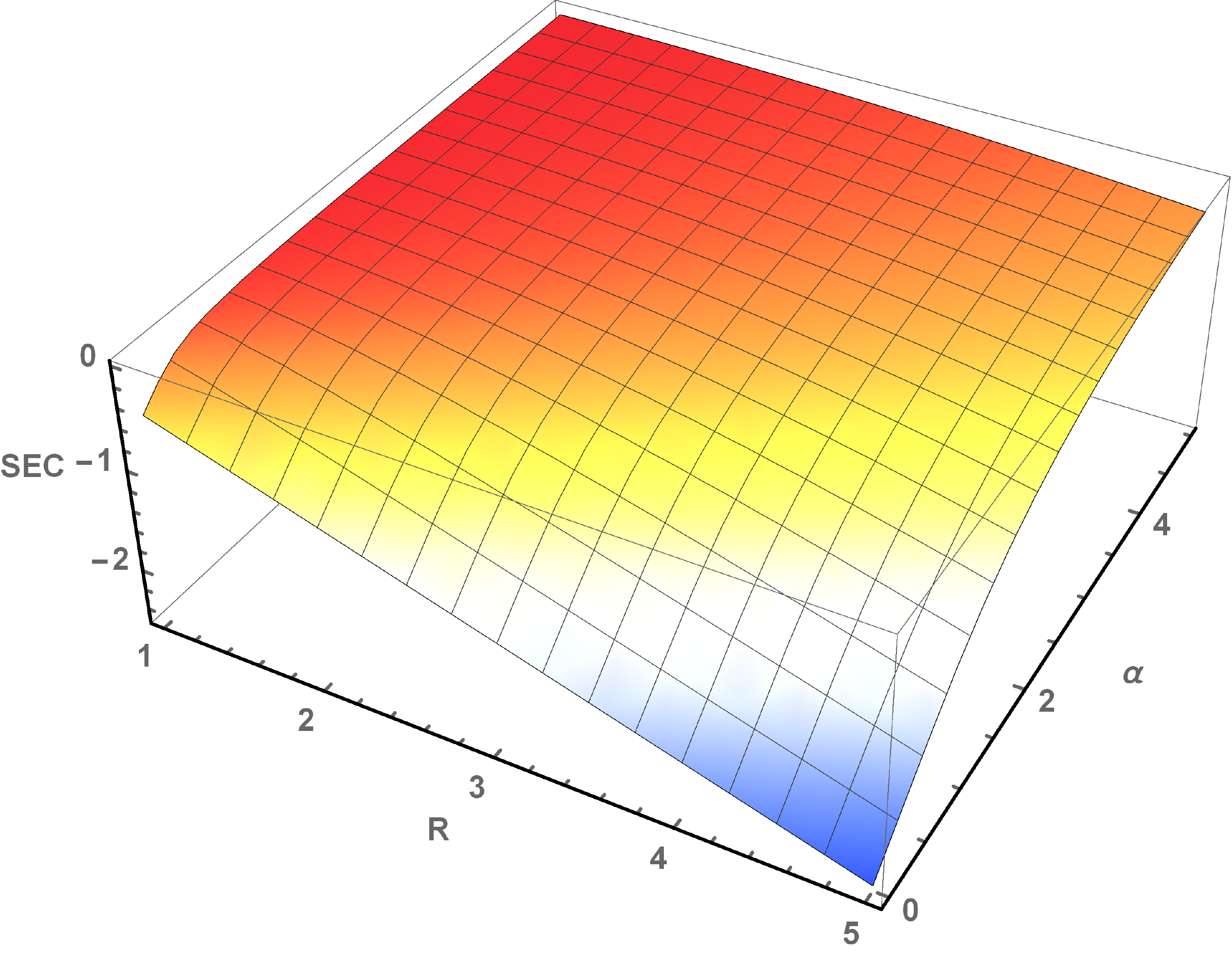}

\caption{Evolution of SEC with respect to $R$ and $\alpha$ for model-1.}
\end{center}
\end{figure}

\subsection{Model-2: $F(R)=R-\log (\beta  R)$}

In this subsection, we consider $F(R)$ as a logarithmic function of the Ricci scalar $R$ with free parameter $\beta>0$. For this $F(R)$, we can write the energy density and pressure as follows
\begin{equation}\label{25}
\rho=\frac{1}{4} (-2 \log (\beta  R)+R+1),
\end{equation}

\begin{equation}\label{26}
p=\frac{1}{4} (2 \log (\beta  R)-R-1).
\end{equation}

Using Eqn. \eqref{25} and \eqref{26}, we present the profiles of energy density $\rho$, DEC, and SEC in Fig. 4, 5, and 6, respectively. We observe that the energy density and DEC are satisfied from those figures, whereas SEC is violated. Moreover, WEC and NEC are also satisfied as $\rho+p$ reduces to zero for this formulation. Nonetheless, the equation of state parameter $\omega$ takes its value as $-1$. These results indicate the accelerated expansion of the universe and compatible with the $\Lambda$CDM model \cite{planck/2018}.

\begin{figure}[H]
\begin{center}
\includegraphics[scale=0.4]{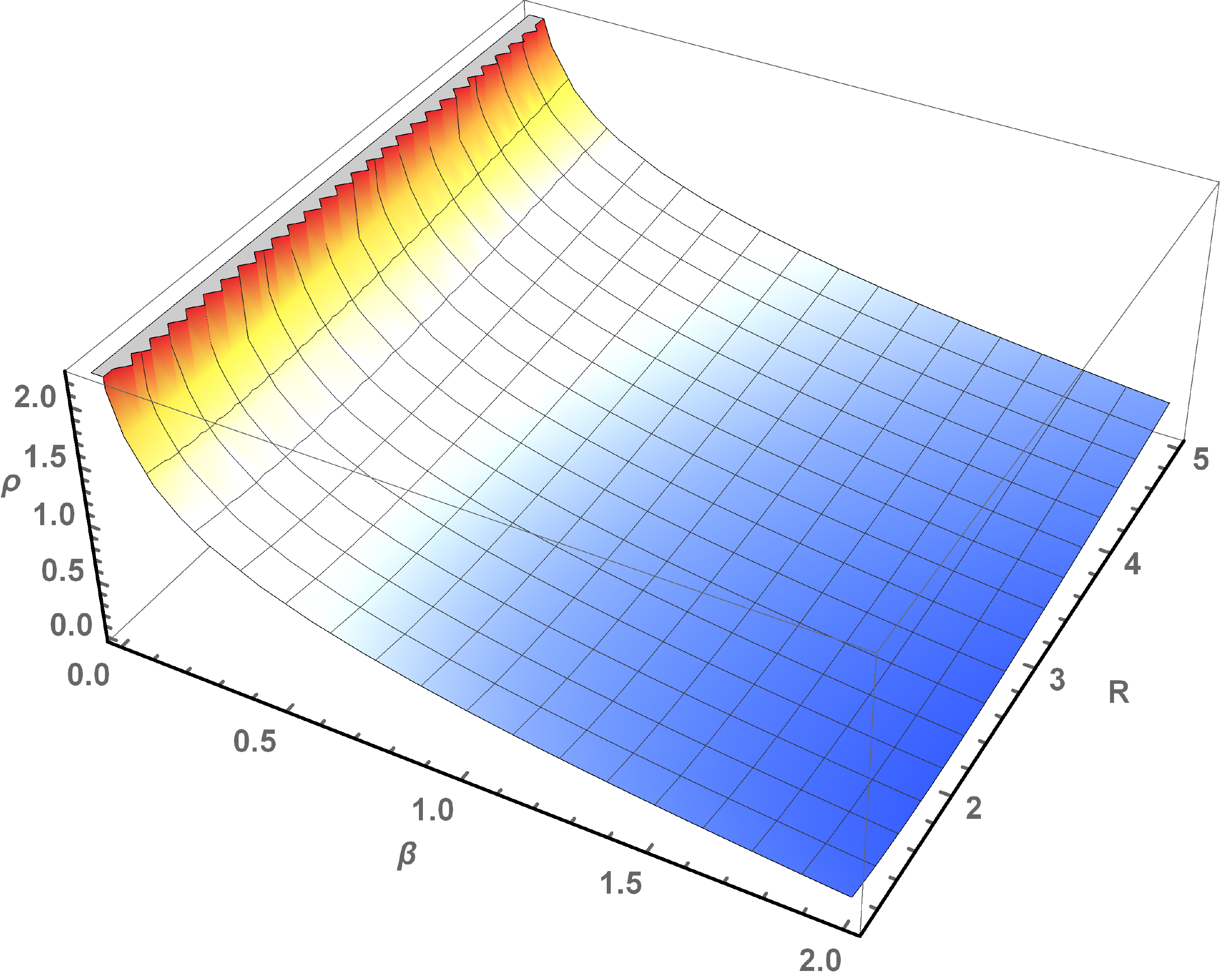}

\caption{Evolution of energy density with respect to $R$ and $\beta$ for model-2.}
\end{center}
\end{figure}
\begin{figure}[H]
\begin{center}
\includegraphics[scale=0.4]{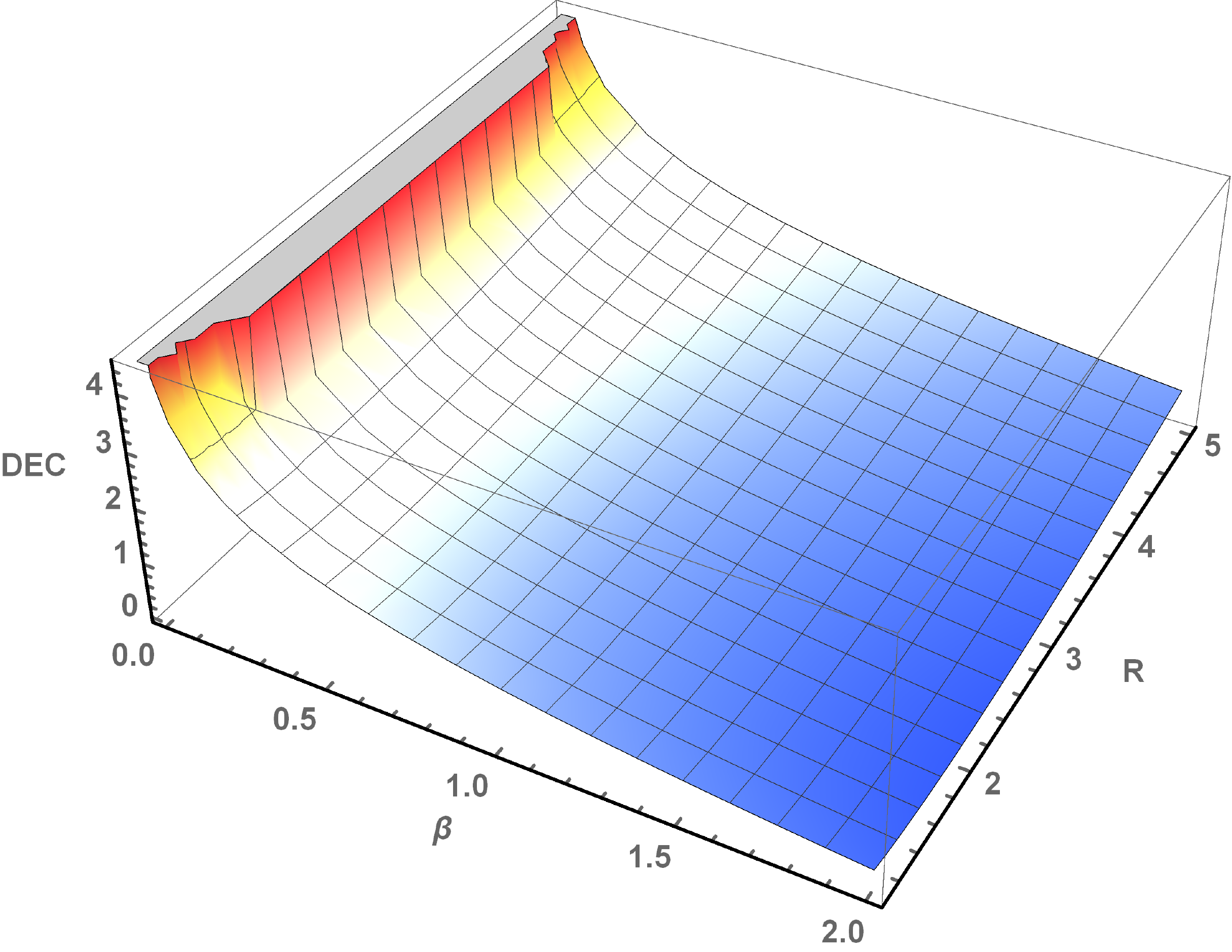}

\caption{Evolution of DEC with respect to $R$ and $\beta$ for model-2.}
\end{center}
\end{figure}
\begin{figure}[H]
\begin{center}
\includegraphics[scale=0.4]{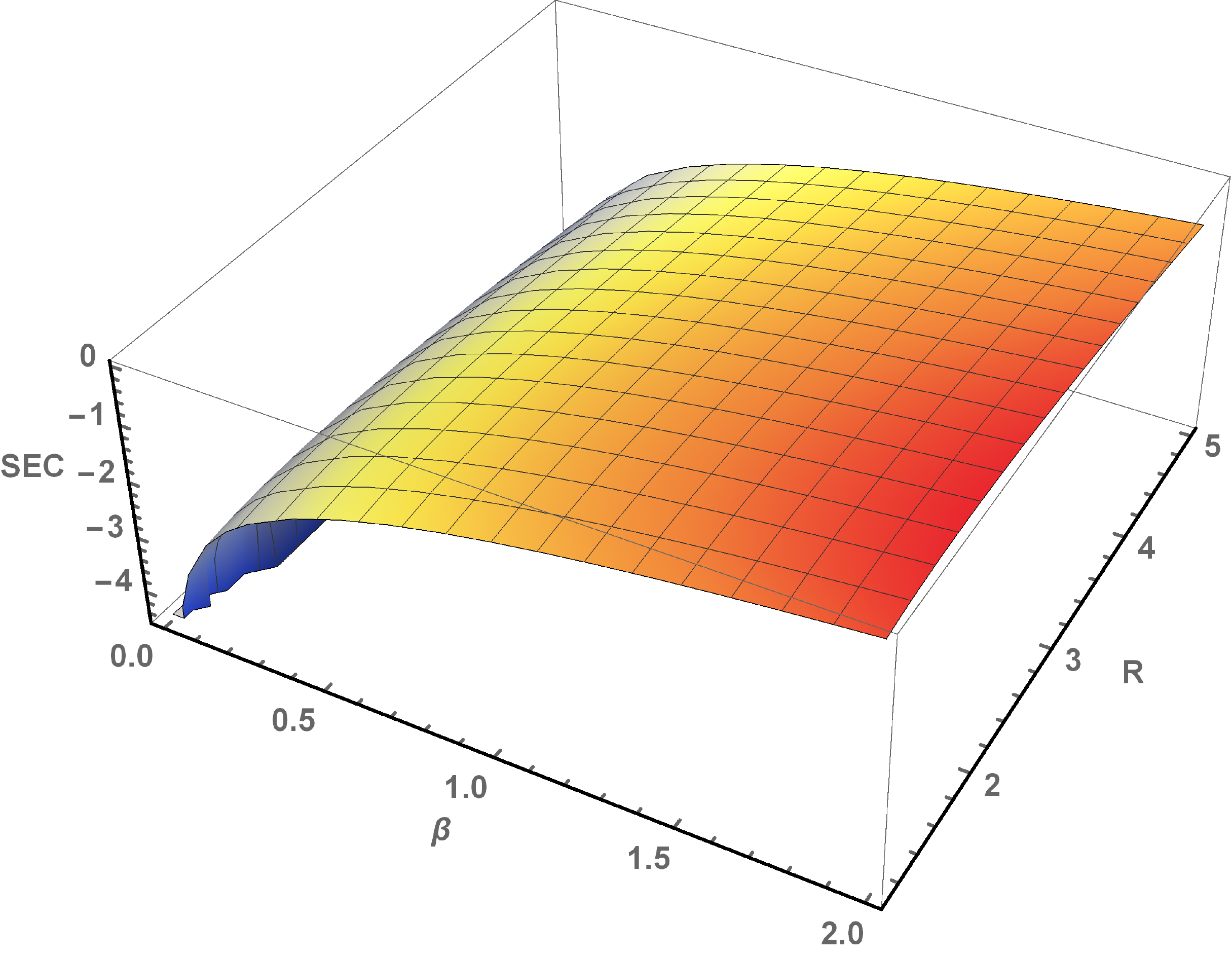}
\caption{Evolution of SEC with respect to $R$ and $\beta$ for model-2.}
\end{center}
\end{figure}

\section{\textbf{Projectively harmonic spacetime}}
In this section, instead of projective flatness, we impose only vanishing projective divergence, namely, $\nabla^hP_{hijk}=0$, into the spacetime geometry to see how this weaker restriction affects the theories of gravity. 

Putting $n=4$ in (\ref{proj}) and then taking divergence we obtain same equation as (\ref{v1}) for a divergence free projective curvature. Thereafter, proceeding similarly we get the results as in equations (\ref{v2}) and (\ref{r}).
Since the Ricci scalar $R$ is constant, we take covariant derivative of (\ref{fr}) and use (\ref{v2}) to produce
\be\nabla_iT_{jk}-\nabla_jT_{ik}=0.\label{codazzi}\ee
Hence we get the following result:
\begin{theorem}\label{thm2}
In a projectively harmonic spacetime solution of $F(R)$-gravity theories, the stress-energy tensor is of Codazzi type.
\end{theorem} 

Now, covariant derivative of (\ref{pf}) gives
\begin{align}
\nabla_kT_{ij}
=&\nabla_kp g_{ij}+(\nabla_k\rho+\nabla_kp)u_iu_j+(\rho+p)\{\nabla_ku_iu_j+u_i\nabla_ku_j\} \notag\\
=&\nabla_kp h_{ij}+\nabla_k\rho u_iu_j+(\rho+p)\{\nabla_ku_iu_j+u_i\nabla_ku_j\} ,\notag
\end{align}
which due to (\ref{codazzi}) produces
\begin{align}\label{eqn:2}
\nabla_kp h_{ij}+&\nabla_k\rho u_iu_j
    +(\rho+p)\{\nabla_ku_iu_j+u_i\nabla_ku_j\}  \notag		\\	
=&\nabla_jp h_{ik}+\nabla_j\rho u_iu_k
    +(\rho+p)\{\nabla_ju_iu_k+u_i\nabla_ju_k\}.
\end{align}
Transvecting (\ref{eqn:2}) with $u^i$ gives
\begin{align}\label{eqn:10}
\nabla_k\rho u_j+(\rho+p)\nabla_ku_j=&\nabla_j\rho u_k+(\rho+p)\nabla_ju_k.
\end{align}
Substituting this back into (\ref{eqn:2}) gives
\begin{align}\label{eqn:20}
\nabla_kp h_{ij}+(\rho+p)\nabla_ku_iu_j=&\nabla_jp h_{ik}+(\rho+p)\nabla_ju_iu_k.
\end{align}
Transvecting (\ref{eqn:20}) with $h^{ij}$ and  $h^{kl}$ yields $2h^{lk}\nabla_kp=0$, implying 
\begin{align}\label{eqn:30}
\nabla_kp=-\dot pu_k;\qquad (\dot p=u^i\nabla_ip).
\end{align}
Transvecting (\ref{eqn:20}) with $u^k$, with the help of (\ref{eqn:30}), we obtain
\begin{align}\label{eqn:40}
(\rho+p)\nabla_ju_i=-\dot ph_{ji}-(\rho+p)u_j\dot u_i; \qquad (\dot u_i=u^k\nabla_ku_i).
\end{align}
By using (\ref{eqn:10}) and (\ref{eqn:40}), we obtain 
\begin{align*}
\nabla_k\rho u_j-(\rho+p)u_k\dot u_j=&\nabla_j\rho u_k-(\rho+p)u_j\dot u_k.
\end{align*}
Transvecting this equation with $u^j$ and $h^{lk}$ gives
\begin{align}\label{eqn:50}
h^{lk}\nabla_k\rho=-(\rho+p)\dot u_kh^{lk}=-(\rho+p)\dot u^l.
\end{align}
The trace of (\ref{fr})
\begin{align}\label{eqn:60}
F_R(R)R=2F(R)+\kappa^2(3p-\rho),
\end{align}
from (\ref{eqn:30}) and (\ref{eqn:50}) gives 
\[
\kappa^2(\rho+p)\dot u^l=-\kappa^2h^{lk}\nabla_k\rho =h^{lk}\nabla_k(F_R(R)R-2F(R)-3\kappa^2 p)=0.
\]
If $\rho+p=0$, then NEC satisfied. Furthermore,
NEC describes the null geodesic congruences, and it is convergent in sufficiently small neighborhoods of every point of spacetime. The physical interpretation of NEC is that particles following null geodesics will observe that gravity tends locally to be attractive (or at least not repulsive) when acting on nearby particles also following null geodesics.

If $\rho+p\neq0$, then it follows from (\ref{eqn:40}) and (\ref{eqn:60}) that 
 \begin{align}\label{eqn:70}
\nabla_ju_i=-\frac{\dot p}{\rho+p}h_{ji}
=\frac{-\kappa^2\dot \rho}{F_R(R)-2F(R)+4\kappa^2\rho}h_{ji}.
\end{align}
Since $h^{lk}\nabla_k\rho=0$, the spacetime is a GRW-spacetime with a scale factor $a(t)$.
Since every GRW-spacetime had the property
\[
\nabla_ju_i=\frac{\dot a}{a}h_{ji}.
\]
comparing with (\ref{eqn:70}) gives
\[
\frac{4\dot a}{a}=\frac{-4\kappa^2\dot \rho}{F_R(R)-2F(R)+4\kappa^2\rho},
\]
which means that 
\[
a=\left(\frac{A}{F_R(R)-2F(R)+4\kappa^2\rho}\right)^{1/4},
\]
where $A$ is a constant.
Hence we get the following results.
\begin{theorem}\label{thm3}
A projectively harmonic perfect fluid spacetime solution of $F(R)$-gravity theories is either a GRW-spacetime
with scale factor 
\[
a=\left(\frac{A}{F_R(R)-2F(R)+4\kappa^2\rho}\right)^{1/4},
\]
or 
$p+\rho=0$.
We have discussed previously that $p+\rho$ is represented as NEC and describes the null energy congruence of the spacetime. Moreover, it describes the nature of gravity in the neighbourhoods of each spacetime point. Primarily, it describes the attractive nature of gravity in the locality of each point.

\end{theorem}
The profile of scale factor for two models are discussed below.
For model-1, the scale factor reads
\begin{equation*}
a=\sqrt[4]{-\frac{A e^{R/\alpha }}{(R-1) \left(e^{R/\alpha }-1\right)}}, \,\ R\neq 1,
\end{equation*}
and, for model-2, the scale factor reads
\begin{equation*}
a=\sqrt[4]{-\frac{A R}{(R-1)^2}}, \,\ R\neq 1.
\end{equation*}
From Fig. 7 and 8, one can see that the value of scale factor decreases as the Ricci scalar increases, which indicates that for higher curvature the scale factor of our universe is small and vice-verse. Also, it suggests that curvature reduces as time goes on  and scale factor increases.
\begin{figure}[H]
\begin{center}
\includegraphics[scale=0.4]{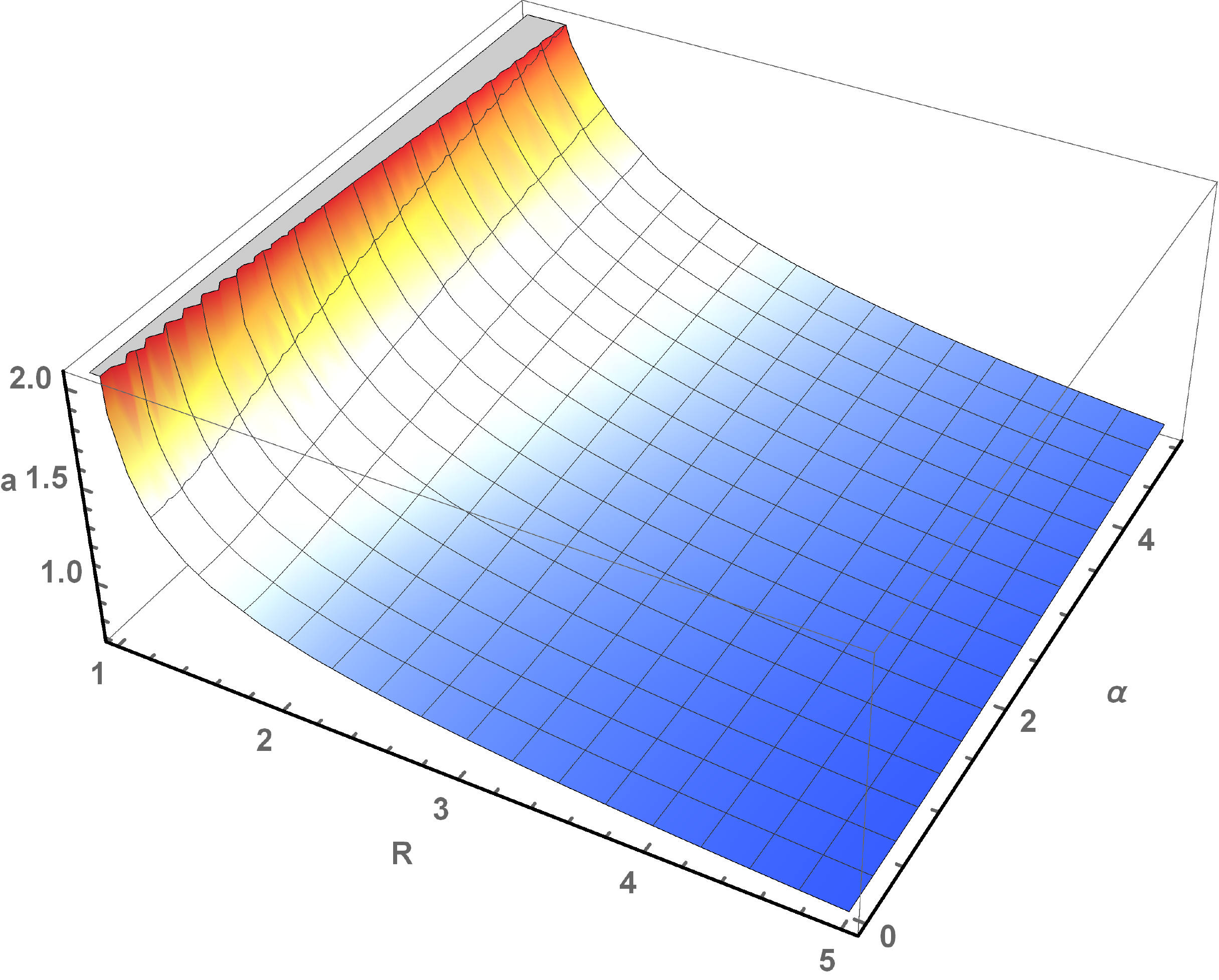}

\caption{Evolution of scale factor with respect to $R$ and $\alpha$ with $A<0$ for model-1.}
\end{center}
\end{figure}
\begin{figure}[H]
\begin{center}
\includegraphics[scale=0.4]{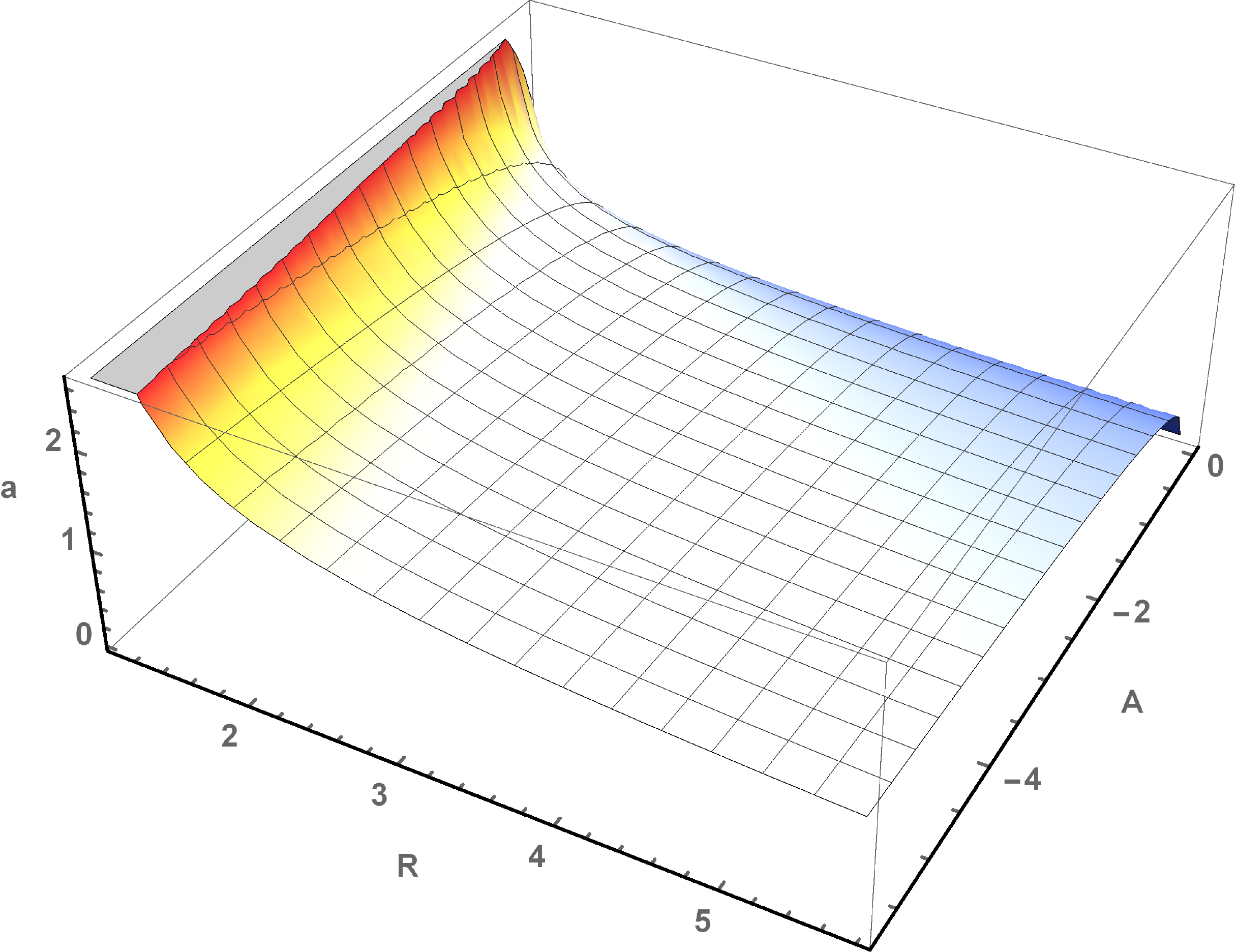}
\caption{Evolution of scale factor with respect to $R$ and  $A<0$ for model-2.}
\end{center}
\end{figure}

\begin{remark}
In \cite{beesham}, the authors rightly showed that conharmonic flatness in a spacetime constrains the standard theory of gravity to a trivial vacuum case with $R_{ij}=0$ and $R=0$ case. Even in modified gravity theories the situation does not improve much, it ends up resulting into vanishing effective pressure and energy density.

We further can add to the findings of the article \cite{beesham} that instead of conharmonic flatness if one imposes divergence free conharmonic curvature into the geometry, then the Ricci scalar $R$ of the spacetime is constant and the Ricci curvature tensor is of Codazzi type. Therefore, same conclusions as in Theorem \ref{thm2} and Theorem \ref{thm3} follow from similar arguments.   
\end{remark}
\begin{remark}
All the results of this section are also true for the standard theory of gravity governed by Einstein's field equations (\ref{efe}).
\end{remark}

\section{\textbf{Discussion}}
This manuscript has focused on studying the projectively flat spacetime in the framework of $F(R)$ gravity. We have analyzed our study both analytically and graphically. We have used the analytical method to construct our formulation, and graphical analysis has been used to check the stability of two cosmological toy models for a better understanding, such as $F(R)=R-\alpha  \left(1-e^{-\frac{R}{\alpha }}\right)$ and  $F(R)=R-\log (\beta  R)$.

Moreover, we have examined the stability analysis of the cosmological models using energy conditions. For model-1, we have presented the profiles of energy conditions in Figures 1,2 and 3. It has been shown that the energy density is positive throughout the evolution of parameters $R>1$ and $\alpha>0$. Furthermore, DEC, NEC and WEC have been satisfied, whereas SEC violated. All of the above profiles of energy conditions are in agreement with the accelerated expansion of the universe. In addition, the equation of state parameter $\omega$ takes its value as $-1$, which also indicates the dark energy era. Nevertheless, these results are also compatible with the $\Lambda$CDM model. Similarly, for the second model, we have depicted all the energy conditions in Figures 4, 5, and 6. The results we have obtained for model-2 are aligned with the results of model-1.

Finally, we have studied the projectively harmonic spacetime in the framework of $F(R)$ gravity. We have formulated the scale factor in terms of $F(R)$ and alternately found $\rho+p=0$. Then we have presented the evolution of the scale factor for two models concerning Ricci scalar $R$ and model parameter in Fig. 7 and 8, respectively. From those figures, we have observed that for large values of $R$, scale factor is small, which suggests that the scale factor of the universe is small in the early phase due to the high curvature, and later it increases. Alternately, NEC has been satisfied, that is, $\rho+p=0$. The advantage of this setup is that it does not allow hypothetical astronomical objects like wormholes (where NEC should be violated) and allow the Big-Bang singularity.

\section*{Acknowledgments}

A.D. and L.T.H. are supported by the FRGS research grant (FRGS/1/2019/STG06/UM/02/6). S.M. acknowledges Department of Science \& Technology (DST), Govt. of India, New Delhi, for awarding Junior Research Fellowship (File No. DST/INSPIRE Fellowship/2018/IF180676). We are very much grateful to the honorable referee and to the editor for the illuminating suggestions that have significantly improved our work in terms of research quality, and presentation.

\end{document}